\documentclass[a4paper, 3+]{elsarticle}

\usepackage{hyperref}
\usepackage[margin=1in]{geometry}
\usepackage{latexsym, amssymb, amsmath, amsthm, graphics, graphicx, subfigure, bbm, stmaryrd}%\usepackage{amsfonts}
\usepackage{epstopdf}
\usepackage{mathtools}
\usepackage{tikz}
\usetikzlibrary{positioning,decorations.markings,arrows}
\tikzset{state/.style={circle,draw=blue}}
\newtheorem{theorem}{Theorem}

\newtheorem{corollary}[theorem]{Corollary}

\newtheorem{definition}[theorem]{Definition}
\newtheorem{example}[theorem]{Example}

\newtheorem{proposition}[theorem]{Proposition}
\newtheorem{remark}[theorem]{Remark}

\newtheorem{assumption}[theorem]{Assumption}

\newcommand{\field}[1]{\mathbb{#1}}
\newcommand{\R}{\field{R}}

\newcommand{\E}{\field{E}}

\newcommand{\norm}[1]{\left\lVert#1\right\rVert}
\newcommand{\ms}{\bar \sigma}
\newcommand{\veps}{\varepsilon}
\begin{document}

\begin{frontmatter}

\title{Broken Detailed Balance and Non-equilibrium Dynamics in Noisy Social Learning Models}	


\author[SUTD]{Tushar Vaidya}
\address[SUTD]{Singapore University of Technology and Design, 8 Somapah Road, Singapore 487372}
\ead{tushar\_vaidya@sutd.edu.sg}

\author[Chul]{Thiparat Chotibut\corref{mycorrespondingauthorTC}}
\cortext[mycorrespondingauthorTC]{corresponding author}
\address[Chul]{Department of Physics, Faculty of Science, Chulalongkorn University, 254 Phayathai Road, Pathumwan, Bangkok 10330 Thailand}
\ead{thiparatc@gmail.com, thiparat.c@chula.ac.th}

\author[SUTD]{Georgios Piliouras}
%\address[SUTD]{SUTD}
\ead{georgios@sutd.edu.sg}

\begin{abstract}
We propose new Degroot-type social learning models with noisy feedback in continuous time. Unlike the standard Degroot framework, noisy information frameworks destroy consensus formation. On the other hand, noisy opinion dynamics converge to the equilibrium distribution that encapsulates correlations among agents' opinions. Interestingly, such an equilibrium distribution is also a non-equilibrium steady state (NESS) with a non-zero probabilistic current loop. Thus, noisy information source leads to a NESS at long times that encodes persistent correlated opinion dynamics of learning agents. Our model provides a simple realization of NESS in the context of social learning. Other phenomena such as synchronization of opinions when agents are subject to a common noise are also studied. 
\end{abstract}


\end{frontmatter}
\section{Introduction}
A central quest in the field of opinion dynamics is to understand how opinion exchange among individuals in a social network can give rise to emergent social phenomena such as consensus formation, polarization, and fragmentation of opinions.
Due to advances in information technology and accessibility to massive social data, empirical studies of opinion dynamics have received increasing attention and provide insights into developing quantitative models of opinion dynamics \cite{Onnela:kb,schinckus2018ising,Leskovec:ta, simmons2011memes, FernandezGracia:2013kz,Das:ef, masuda2010collective}.
Even before the social media age, many theoretical models had been proposed to explain social phenomena.
Among the most studied ones are the Voter model \cite{Clifford:s-r5Vq0b,Holley:we} and its variants \cite{Castellano:2009bt, sood2008voter}, which describe dynamics of discrete opinions in a population, such as electoral votes for political parties, through the lens of interacting particle systems \cite{liggett1985interacting}. 
Although these models are highly idealized, tools from discrete classical spin models in statistical mechanics can be employed to elucidate that consensus formation is a collective phenomenon, similar to how a magnetic ordered phase can emerge from microscopic interactions among discrete-valued spins in magnetic systems \cite{Castellano:2009bt,krapivsky2010kinetic}. Despite a clear connection to discrete-spin systems, Ising-like models of opinion dynamics are not appropriate when opinions are real valued and continuous, especially in financial markets.  There, information is also exchanged simultaneously, not just pairwise.

%---- P2 ------
In this work, we will focus on a class of models that concern the dynamics of continuous opinions, such as traders' beliefs on a stock value. The Econophysics literature tackles slightly different aspects of financial markets \cite{chakraborti2011econophysics,chakraborti2011econophysicsII}; here, we will not deal with limit order books which is a central theme in Econophysics \cite{abergel2011econophysics}. Instead, we will focus on how the beliefs of traders arise through microscopic interaction and learning. Stanley et. al \cite{stanley1996anomalous} summarized the  need to incorporate microscopic interactions to understand macroscopic behaviours in Econophysics. Our work follows this theme, investigating how microscopic interaction and learning give rise to collective phenomena in financial markets. 

One of the paradigmatic continuous opinion models that incorporate interaction is DeGroot's repeated linear updating model \cite{Degroot:1974aa}, see also the survey in \cite{masuda2017random}. In Degroot's model of continuous opinion dynamics, each agent $i \in \{ 1, 2, \dots, n \}$ holds a real-valued opinion $X_t^i \in [0,1]$ or in some compact domain. At each time step $t$, agents simultaneously update their opinions by taking the weight average of others' opinions: 
\begin{equation}\label{eqn: degroot}
X_{t+1}^i = \sum_{j=1}^{n} A_{ij} X_t^j,
\end{equation}
where $A_{ij}$ is a right stochastic matrix with non-negative entries. Its rows sum to one $\sum_{j=1}^{n} A_{ij} = 1$. An entry $A_{ij}$ shall be interpreted as an influence of agent $j$'s opinion on agent $i$, or the degree of trust that agent $i$ has on agent $j$'s opinion.  
This mathematical setup is equivalent to a Markov chain defined on a directed graph, where a non-zero entry of the influence matrix $A_{ij} > 0$ specifies a directed edge from node $j$ to node $i$. 
Depending on the social graph structure encoded in the influence matrix $A_{ij}$, different dynamical scenarios could arise \cite{Degroot:1974aa,golublearning,Golub:er,mossel2017opinion, liu2014control}. 
For instance, in a strongly connected social network ($A$ is irreducible), convergence to a consensus is guaranteed (i.e., there exists  $X^*$ such that $\lim_{t\rightarrow\infty}X_t^i = X^*$ for every agent $i$) if and only if the network is aperiodic \cite{Golub:er}.  
Generalizations of DeGroot's model that account for a social assumption of bounded confidence, where only agents with nearby opinions interact, can also lead to consensus formation \cite{Castellano:2009bt}. Bounded-confidence models such as Deffuant model \cite{Deffuant:2000gi} and Hegselmann-Krause model \cite{Hegselmann:th, lorenz2007continuous} also exhibit polarization (formation of two opinion clusters are formed), fragmentation of opinions (formation of many opinion clusters), and possess rich dynamical behaviours due to their non-linear update rules \cite{Castellano:2009bt}. 

In many realistic situations, experimental studies reveal that individuals tend to adopt Degroot-like learning \cite{becker2017network, chandrasekhar2015testing}. Although Bayesian learning models with priors are alternatives to describing  learning rules, the computation is more cumbersome as the agents need to keep track of both their prior distributions and update their posteriors. For simplicity, we will study opinion dynamics that arise from DeGroot-type update, with the focus on the role of noise inherent in the information source on the opinion dynamics of the entire populations.  Note that noisy opinion dynamics models are also getting increasing attention in the control and engineering literature \cite{askarzadeh2019stability}.
\paragraph{Our Contribution}
Using tools from It\^o calculus, we introduce a noisy information source into Degroot-type opinion dynamics models. We show that although noisy information source destroys consensus formation, %in standard Degroot-type update, 
the dynamics eventually converge to the equilibrium (stationary) distribution which encapsulates correlations among agents' opinions at long times.  The main finding is that such equilibrium distribution is also a non-equilibrium steady state (NESS). We study the concept of NESS numerically and analytically in the scenario of two learning agents subjected to a common noise source. To our knowledge, this is the first realization of NESS in a simple noisy social learning model. We also report an interesting phenomenon of opinion synchronization when agents are subject to a common noise, and sketch the criterion for opinion synchronization at long times. 

\section{Organization of the Paper}
In order to utilize the tools from stochastic processes to analyse noisy social learning, we first convert Degroot-type learning dynamics from discrete-time to continuous-time, and discuss the consensus formation criterion in \autoref{s: noiseless}.  An explicit method is demonstrated that makes use of the link between discrete-time and continuous-time Markov chains. 

Noisy dynamics are introduced in \autoref{s: noisy}. White noise is added and the It\^o interpretation is taken. In \autoref{s: stationary distribution}, we study the Fokker-Planck equation and how it converges to an asymptotic measure following techniques from \cite{bolley2012convergence}.
 %The techniques are borrowed from a recent paper in mathematics \cite{bolley2012convergence}. 
  While existing methods are known to answer this question, the method we show presents an alternative argument. 
%Also in this section, the case for non-symmetric interaction is sketched and will be referred to in the appendix. The symmetric case has an easy interpretation and can be solved using different methods. The non-symmetric case presents new challenges.  Symmetry here means that agents put the same weight on the opinion as their counterpart. In the the discrete model it means $A_{ij}=A_{ji}$.   

Section \ref{s: synchronization} discusses a combination of diverse ideas and tools. If the number of agents is different from the number of Brownian drivers then there is no guarantee that the resulting transition density of the Fokker-Planck equation is regular.  The notion of a Kalman controllability matrix is introduced. Control theory make the analysis easier.  With just one source of uncertainty, the system can still have a proper probability density in $\R^n$ depending on how agents learn. Interestingly, with one common noise source, we demonstrate that some of the agents' opinions can synchronize, and the correlated opinion dynamics live in a lower dimensional linear embedding of dimension $k<n$. We attempt to establish the condition for dimensionality reduction using tools from Control theory. Finally, we study the formation of non-equilibrium steady states (broken detailed balance). For the case of only two agents, an analytical formula is known for the long-time limit distribution, and \autoref{s: ness} expands on a concrete example.

\section{Consensus Formation in Degroot-like Learning with Feedback}
\label{s: noiseless}
We first introduce motivations for the deterministic model and review results on consensus formation in this section.
 How do individuals learn when observing others in financial markets? Canonical models suggest simple rules: agents update their own beliefs after observing others. When interacting and learning, players act simultaneously or take turns. Previous models treat learning and interaction to be the same process. But more realistically, we treat these two concepts independently, as we now discuss. 

\subsection{Learning in Markets}
Market prices are the result of interaction among agents. We are interested in learning and interacting models for traders in financial settings. If prices already reflect all available information, then there is no incentive for participants to interact and learn from each other. But the markets are not static and thus individuals are always willing to acquire new information and learn from their environments. Models in mathematical finance postulate a model for a market price such as Geometric Brownian Motion. But this is just a model of the world. Market prices arise because individuals interact and learn from the world. 

Given the fragmented nature of markets, it is never the case that one price incorporates all available information.  Individuals have dispersed beliefs around a central equilibrium value, hence the existence of trading. People will have different opinions on an asset's value or financial instrument's price. Until full revelation or maturity of the instrument \footnote{Bonds and Derivatives have finite lives, that is they mature.}, live markets allow individuals to posit different prices and beliefs. Learning through communication is a reflection of reality. Today's markets are not only electronic but also very tightly linked globally. Updating of one's beliefs of an asset's price occurs not in isolation but in the presence of several competing trading platforms and market intermediaries. One trader's actions or price updates are observed by others. 

\subsection{Learning versus Interaction: Belief Update}
\emph{Interaction} in our model has a specific meaning. It means how individuals observe actions or beliefs of others and update their own opinions. This is the physical process. Interaction can be represented as an approximation to reality. Networks or matrices represent this in compact form based on the network topology. 

\emph{Learning} in our case is separate from interaction. 
%Suppose we look at an individual who updates her opinion based on opinions of others in her own sub network. 
The \emph{learning} part  encodes how %decodes feedback for agent. How
 skillful each agent is  in ascertaining the quality of feedback. Our model, which is based on Ref. \cite{vaidya2018learning}, accounts for an external  feedback, where $\bar \sigma$ is the proposed truth. It influences the opinions of an entire population that can converge to $\bar \sigma$ under appropriate conditions of the interaction matrix $A$ in \eqref{eqn: degroot}.  Another interpretation of $\ms$ is that agents in one market also observe prices being quoted in another location. For example, stocks are not solely traded on NASDAQ or NYSE. So while the main market may be one of the two, trading activity can be any other trading platform and geographic location. Sometimes the best price to buy/sell is not through the major trading platform but  an alternative venue. Fragmentation of markets, even exchange based ones, is real and increasing \cite{foucault2013market}. 
\newline

Let us setup the model in more detail. Let $A$ be a row-stochastic matrix, $\bar \sigma$ be the equilibrium value times the column vector $(1,1, \dots, 1)$, and $\mathcal{E} \equiv \text{diag}(\varepsilon_1, \varepsilon_2, \dots, \varepsilon_n)$ be the diagonal matrix of learning rates that are non-negative; we first recall the discrete-time update from \cite{vaidya2018learning}; for the opinion vector $X_t \in \mathbb{R}^n$, the dynamics are
\[
X_{t+1}= \underbrace{A X_t}_{interaction} + \underbrace{\mathcal{E}( \bar \sigma - X_t)}_{learning}.
\]
The learning rates represent each agent's unique ability. Vaidya et al. \cite{vaidya2018learning} established exponentially fast convergence to $\bar \sigma$ for the above model in discrete time under appropriate criterion on $A$. What about the large-time limit in a continuous-time model? To study this limit, we need to go from discrete-time dynamics to the continuous-time.  Let us rewrite the above equation to account for the fact that time is not moving forward in unit steps but in units of $dt$ and subtract $X_t$ from both sides. The learning rates are now also with respect to time increment $dt$. Let us denote $\bar{\sigma}=\bar{\sigma}\begin{pmatrix} 
1 \\  \vdots \\ 1\end{pmatrix} $ for brevity. Then, subtracting both sides of $X_{t+dt}= A X_t + \mathcal{E}dt( \bar \sigma - X_t)$ by $X_t$ and dividing by $dt$, we obtain
\begin{equation}\label{eqn: dis2cont}
\dfrac{X_{t+dt}-X_t}{dt}= \dfrac{(A-I)}{dt}X_t  +\mathcal{E}( \bar \sigma - X_t).
\end{equation}
From continuous time Markov chain theory, we know that 
\[
\lim_{dt\rightarrow0}  \dfrac{(A-I)}{dt}=Q 
\]
is an infinitesimal generator matrix whose rows sum to zero.  Although $Q$ can be infinite for an infinite number of states, we restrict attention to the finite case.  By taking the limit $dt \rightarrow 0$ of \eqref{eqn: dis2cont}, we obtain the continuous time model 

\begin{equation}
\dot{X_t}=Q X_t + \mathcal{E}( \bar \sigma - X_t).
\end{equation}

\begin{definition}\cite{suhov2008probability} \label{def: CTMCQmatrix}
	The properties of a valid $Q$ on a finite or countable set $I$ are for all $i,j \in I$:
	
	\begin{enumerate}
		\item Non-positive diagonals $Q_{ii} \leq 0$,
		\item Non-negative off-diagonals $Q_{ij} \geq 0 \mbox{ for } i\neq j$,
		\item Zero row sum $Q_{ii}=-\sum_{j \neq i} Q_{ij}$ or $\sum_j Q_{ij}=0$.
	\end{enumerate}
\end{definition}

Definition \ref{def: CTMCQmatrix} requires a bit more explanation. In discrete time, the agent was averaging her opinion with respect to everyone else, including her own through weights matrix $A$. Once we pass to continuous time, the interpretation is slightly different.  
\begin{example}
	Consider $\mathcal{E}=0$, so there is no learning or feedback and examine agent $i$'s dynamics are given by ODE
	\[
	\dot{X}_{t}^i = \sum_{j=1}^{n} Q_{ij} X_t^j=\sum_{j\neq i}^{n} Q_{ij} (X_t^j -X_t^i).
	\]
	
	It is easy to see that by property $3$ of definition \ref{def: CTMCQmatrix}, agent $i$ is responding infinitesimally (with respect to time) her opinion in response to how far off she is from the other agents with weight $Q_{ij}$.  
\end{example}
Now if $\mathcal{E}\neq 0$, then this additional term in the dynamics reflects an adjustment to the common news source. Armed with this knowledge of $Q$, we proceed to state the fairly simple result. 

\begin{proposition}
	Suppose the continuous time dynamics are 
	\begin{equation}\label{eqn:ctstimemodel1}
	\dot{X_t}=Q X_t + \mathcal{E}( \bar \sigma - X_t),
	\end{equation}
	where $X_t\in \R^n$, the generator matrix $Q$ has full rank, $\mathcal{E}>0$, and the dot represents time derivative. Then agents in the system eventually reach the consensus
	\[
	\lim_{t\rightarrow\infty} X_t=\bar{\sigma}.
	\]
	
\end{proposition}

\begin{proof} By rearranging \eqref{eqn:ctstimemodel1} as $\dot{X_t}=(Q - \mathcal{E} )X_t + \mathcal{E}\bar \sigma$ and introducing $B\triangleq Q- \mathcal{E}$, the solution $X_t$ can be obtained by directly solving the ODEs using integrating factor as

	\[
	X_t=X_0e^{Bt}+\int_{0}^{t}e^{B(t-s)}\mathcal{E}\bar \sigma ds.
	\]
	In order for the time evolution to remain finite, we need $\displaystyle\lim_{t\rightarrow\infty}X_0e^{Bt}=0$. In fact, for strictly positive learning rates $\mathcal{E} > 0$, all eigenvalues of $B$ are ensured to have a negative real part, and thus $\displaystyle\lim_{t\rightarrow\infty}X_0e^{Bt}=0$. This follows directly from Gershgorin's circle theorem on eigenvalues. No connectedness assumption is made on $A$ or its generator matrix $Q$; the underlying embedded graph could be entirely disconnected. Agents can be totally disconnected in a graphical sense or can learn and interact from each other and be strongly connected. Provided agents are learning, we will have $\displaystyle\lim_{t\rightarrow\infty}X_0e^{Bt}=0$. 
	
	Now consider the second term in the solution. Substituting $y=t-s$ and $dy=-ds$ and making a change of variables yields
	
	\begin{align*}
	\int_{0}^{t}e^{B(t-s)}\mathcal{E}\bar \sigma ds =\int_{0}^{t}e^{By} dy \, \mathcal{E}\bar \sigma.
	\end{align*}
	As $\mathcal{E}\bar \sigma$ is a constant it comes out of the integration.  And as we assumed that $B$ has full rank and so invertible, therefore
	\begin{equation*}
	\int_{0}^{t}e^{By} dy=B^{-1}(e^{Bt}-I).
	\end{equation*}
	By taking the limits $t \rightarrow \infty$, the right hand side converges to $-B^{-1}$. 
	The negative (real part of) eigenvalues of $B$ ensure that $e^{Bt}$ converges to zero as $t \to \infty$. Finally, we get that
	\[
	\lim_{t\rightarrow\infty} X_t=-B^{-1}\mathcal{E}\bar \sigma.
	\]
	In fact, the above limit of $X_t$ is $\bar{\sigma}=\bar{\sigma}\begin{pmatrix} 
	1 \\  \vdots \\ 1\end{pmatrix}$, where we have slightly abused the notation of $\bar{\sigma}$ to stand for the constant vector or the scalar.  Why is the limit $\bar{\sigma}$? Consider, expanding out $B$
	
	\begin{align*}
	-B^{-1}\mathcal{E}\bar \sigma&=-(Q-\mathcal{E})^{-1}\mathcal{E} \bar{\sigma}.
	\end{align*}
	Recall that the rows of $Q$ sum to zero, thus 
	\begin{align*}
	(Q-\mathcal{E})\bar{\sigma}\begin{pmatrix} 
	1 \\  \vdots \\ 1\end{pmatrix}&=-\mathcal{E} \bar{\sigma}\begin{pmatrix} 
	1 \\  \vdots \\ 1\end{pmatrix},
	\end{align*}
	which implies
	\begin{equation*}
	\bar{\sigma}\begin{pmatrix} 
	1 \\  \vdots \\ 1\end{pmatrix} =- (Q-\mathcal{E})^{-1} \mathcal{E} \bar{\sigma}\begin{pmatrix} 
	1 \\  \vdots \\ 1\end{pmatrix}.	
	\end{equation*}
	Or more compactly,
	\begin{align*}
	\lim_{t\rightarrow\infty}X_t=-B^{-1} \mathcal{E} \bar{\sigma}=\bar{\sigma}.
	\end{align*}
	Therefore, without noise, our continuous time dynamics converge to the consensus value.
\end{proof}

\begin{corollary}
	Suppose the original interaction matrix $A$ is the identity matrix, there is no interaction between agents and $\mathcal{E}>0$. Then there is still consensus. 
\end{corollary}
The proof of the corollary is just a special case of the above where $Q=0$ as $A=I$. In this case, the setting reduces to the one-dimensional version and we still have convergence to consensus.  A remark should be made here. In a later section, we may assume, for simplicity, $Q$ to be symmetric to ensure convergence to equilibrium. However, in general, $Q$ can be also non-symmetric and a stationary measure can still exist as long as the real parts of all the eigenvalues of $Q$ are negative.

\begin{example}
	As an illustration consider the matrices
	
	\[
	Q=\left(
	\begin{array}{ccc}
	-3 & 3 & 0 \\                
	2 & -5 & 3 \\
	0 & 2 & -2 \\
	\end{array}
	\right) ,\quad \mathcal{E}=\left(
	\begin{array}{ccc}
	1 & 0 & 0 \\
	0& 1 & 0\\
	0 & 0 & 1 \\
	\end{array}
	\right), \quad \ms=\begin{pmatrix}
	0.25\\ 0.25\\ 0.25
	\end{pmatrix}.
	\]
	\begin{figure}
		\centering
		\includegraphics[scale=.45]{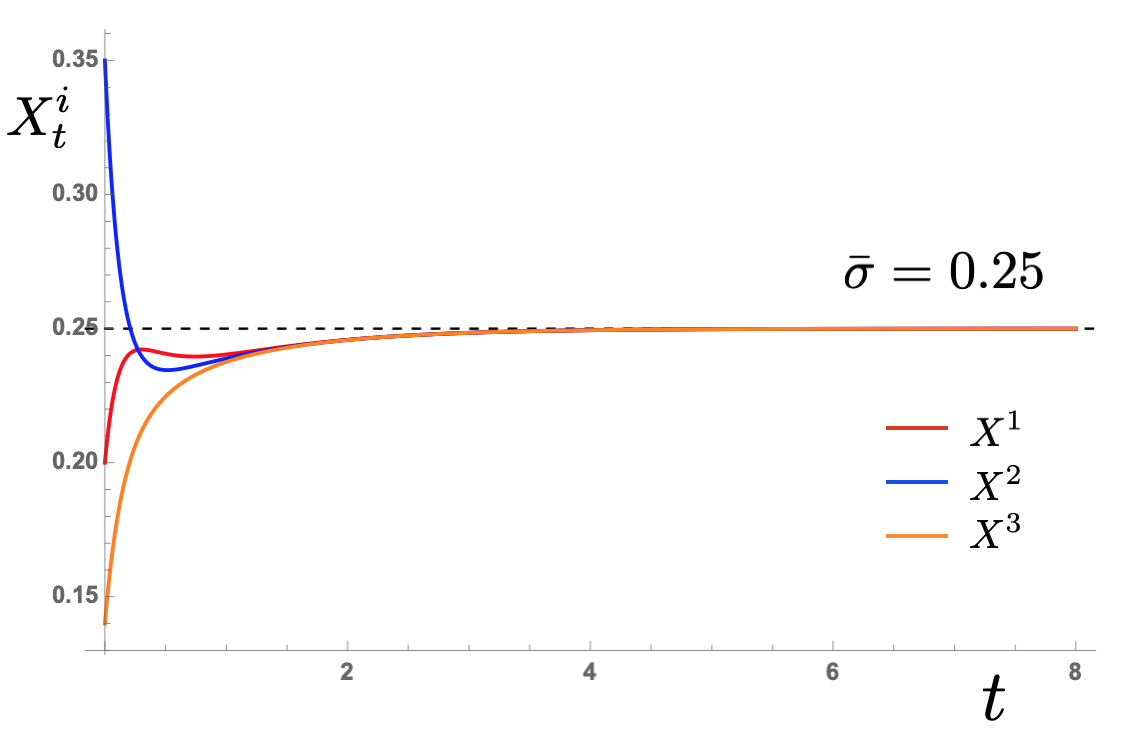}
		\caption{Illustration of consensus formation in continuous-time deterministic opinion dynamics. Three agents start from different initial beliefs but converge quickly to agreement at $\bar\sigma = 0.25$. Agents $1,2 \text{ and }3$ have initial beliefs $0.2, 0.35 \text{ and }0.14$ respectively.}  \label{fig: threeagentsnonoise}
	\end{figure}
Figure \ref{fig: threeagentsnonoise} illustrates the agents quickly reach the consensus in continuous time.
The underlying discrete time Markov chain for the $Q$ matrix is
\[
A=\left(
\begin{array}{ccc}
0 & 1 & 0 \\
\frac{2}{5} & 0 & \frac{3}{5} \\
0 & 1 & 0 \\
\end{array}
\right),
\]
which is irreducible. 
\end{example}
\section{Noise and SDE for Opinion Dynamics with Noisy Feedback}
\label{s: noisy}
As opposed to a static proposed truth $\bar\sigma$ introduced in the previous section studied in discrete-time dynamics in \cite{vaidya2018learning}, in this section we will study the continuous-time dynamics arising from situations in which the proposed truth is unknown and can be uncertain. We will focus on the simple case where the noise in the proposed truth is Gaussian, yet show that interesting non-trivial phenomena can arise. Consider a Gaussian white noise $\Gamma_t$ (random vector) uncorrelated in time with mean zero and independent components, $\, \forall t \geq 0$ 
\begin{align*}
\E[\Gamma_t] &=0, \\
\text{Cov}[\Gamma_t]&=I.
\end{align*}
Adding this noise to $\bar \sigma$, the stochastic dynamics $\dot{X_t}=Q X_t + \mathcal{E}( \bar \sigma + \Gamma_t - X_t)$ can be rewritten as
\begin{equation*}
dX_t=(Q-\mathcal{E}) X_t dt + \mathcal{E}(\bar\sigma + \Gamma_t)dt.	
\end{equation*}

We will focus on how noise affect consensus formation; hence, we will restrict the analysis in the parameter space such that its noiseless dynamics converges to the consensus $\bar\sigma$. Without loss of generality, we can also set $\bar\sigma = 0$. In this case, the stochastic differential equation (SDE) becomes 
\begin{equation}\label{eqn: SDE_WhiteNoise}
dX_t=(Q-\mathcal{E} )X_t dt+ \mathcal{E}dW_t,
\end{equation}
where $W_t=\Gamma_t dt$ is an $n$-dimensional Brownian motion that arises due to a Gaussian white noise in the proposed truth. Note that, mathematically, the formal time derivative does not exist, but we interpret it as such in a distributional sense \cite{evans2012introduction}. Also, we will interpret the noise as an It\^o rather than a Stratonovich representation \cite{yuan2012beyond, mannella2012ito}, for  simplicity in the analysis and for a non-anticipative physical interpretation of learning framework. 

We now briefly discuss the comparison between our stochastic setup and relevant existing literature. Our derivation of the SDE differs from a complementary approach in \cite{galayda2012stochastic}. We are emphasizing the local interactions of opinion dynamics, working from first principles. While the mathematics so far has been simple, we obtained a multidimensional Ornstein-Uhlenbeck process with underlying social interactions of traders. Early work by Follmer et al. \cite{follmer1993microeconomic, horst2005financial} assumed agents fell into three categories and interaction was not explicitly modelled. Rather,  in their model, unlike ours, the process derived was a univariate Ornstein-Uhlenbeck from a discrete series of temporal equilibria as time steps converged to zero. Nevertheless, using microeconomics  showed that even with a limiting \textit{Ornstein-Uhlenbeck process in a random environment}, significant technicalities remain. The resulting one dimensional Ornstein-Uhlenbeck process has random drift and diffusion coefficients. There has been ongoing work in this area \cite{henkel2017agent, henkel2017quantum, pakkanen2010microfoundations}. While agents are heterogeneous, social interaction is not addressed specifically, but analysis focuses on the setting where the number of traders tends to infinity.  For us, we will assume constant drift and diffusion coefficients. If the matrices $Q$ and $\mathcal{E}$ were random or time varying the analysis introduces an undue amount of technicalities \cite{huang2016markov}.

\subsection{Constant drift and diffusion}
The first case to study is when $Q$ and $\mathcal{E}$ are constant. The primary interest is when the noise term is different to all players. The above process can be written (letting $B=Q-\mathcal{E}$ and assuming, as mentioned earlier, the real part of all the eigenvalues are negative to ensure convergence to the equilibrium $\bar \sigma = 0$ in the deterministic case) as
\begin{equation}\label{eqn:model1}
dX_t=BX_t dt+\mathcal{E}dW_t,
\end{equation}
where $\mathcal{E}$ is an $n\times n$ diagonal matrix of learning rates and $dW_t$ is an n-dimensional Brownian motion. Actually, this is a linear SDE because the number of independent Brownian drivers is the same as the number of agents $n$. We could take an $m$-dimensional Brownian motion, where $m<n$. In this case, the regularity of the asymptotic or transition density of $X_t$ is not guaranteed. So we concentrate on the simplest case first to analyse the opinion dynamics.  Note that Ornstein-Uhlenbeck processes are popular models to study a variety of stochastic phenomena in financial markets  and biology \cite{da2013generalized, ramsza2010fictitious, lima2019breaks, kessler2017stochastic}.

For a solution to exist, assume that the initial condition $X_0$ is independent of the filtration generated by the process $X_t$ or $W_t$. The initial condition can be random or given. Let  $\mathcal{F}_0=\sigma{(X_0)}$ be the sigma algebra of $X_0$ and assume that the Brownian motion admissible filtration$\{\mathcal{F}_t\}_{t\geq0}$ is independent of $\mathcal{F}_0$. If $X_0$ is random we assume that $\E[|X_0|^2]<\infty$. An alternative way to say this is that the solution depends on the initial condition $X_0$.  As $B$ and $\mathcal{E}$ are constant, the Lipschitz condition and linear growth bound are trivially satisfied.

The solution of \eqref{eqn:model1} is
\begin{equation}
X_t=X_0 e^{Bt} +\int_{0}^{t}e^{B(t-s)}\mathcal{E}dW_s,
\end{equation}
and can be verified using simple differentiation, where $X_0$ is some fixed initial condition. 

The mean is $\E (X_t) = e^{Bt} X_0$ and the covariance matrix is \newline $\E \left[ \left( X_t -\E X_t\right)\left( X_t -\E X_t\right)^\top  \right]$. The rules of It\^o calculus apply and  
\[
\E [\int_{0}^{t}e^{B(t-s)}\mathcal{E}dW_s]=0,
\]
because it is an Ito integral with constant coefficients. It\^o integrals are Gaussian and this result can be seen in \cite{evans2012introduction}.

Note that that if $E[X_0]=0$, then $\E[X_t]=0$ $\forall t > 0$ and the covariance matrix can be simplified

\begin{align*}
\text{Cov}(X_t)&=\E [X_t\,X_t^\top]\\
&=\E [(\int_{0}^{t}e^{B(t-s)}\mathcal{E}dW_s)(\int_{0}^{t}e^{B(t-s)}\mathcal{E}dW_s)^\top]\\
&\mbox{(by It\^o isometry)}\\
&=\int_{0}^{t}e^{B(t-s)}\mathcal{E}\mathcal{E}^\top e^{B^\top (t-s)}ds.
\end{align*}
In any case, $\lim_{t\rightarrow\infty}\E (X_t)=0$. The influence of any fixed non-zero initial condition diminishes because $B$ is asymptotically stable: its eigenvalues are negative. One thing to note is that we are interested in the covariance matrix and not the autocorrelation 
$
\E \left[ \left( X_t -\E X_t\right)\left( X_s -\E X_s\right)^\top  \right]
$ for different times $s,t$. 

The covariance matrix of the stochastic process $X_t$ is

\begin{equation*}
C_t \triangleq \int_{0}^{t} e^{B(t-s)}\mathcal{E}\mathcal{E}^\top e^{B^\top (t-s)}ds= \int_{0}^{t} e^{By}\mathcal{E}\mathcal{E}^\top e^{B^\top y}dy.
\end{equation*}
The above equality follows from a simple change of variables. The corresponding asymptotic (stationary) covariance matrix is then
\begin{equation}\label{eqn:limitcov1}
C^*\coloneqq \lim_{t\to\infty} C_t =\int_{0}^{\infty} e^{By}\mathcal{E}\mathcal{E}^\top e^{B^\top y}dy.
\end{equation}

%Interestingly, $CC^\top$ is a rank one matrix and the covariance matrix is also a low rank, however proving this can be difficult and there have been many attempts to analyse the low rank of the covariance matrix. One rough way is to calculate an approximate matrix from Gaussian quadrature methods. And we can get within a some small error approximation of the original matrix. So while we postulate that the because of one dimensional noise the covariance matrix could have some small eigenvalues, a numerical approximation, which is of low rank can tell us that indeed approximately this is the case. For a more technical discussion the reader can consult \cite{simoncini2007new}.

Note that \eqref{eqn:model1} can also be described by superscripts denoting components for each row
\[
dX_t^i=\sum_{j=1}^{n}B_{ij}X_t^j dt + \sum_{j=1}^{n}\mathcal{E}_{ij}dW_t^j \quad ; i=1,\cdots,n. 
\]

Following \cite{pavliotis2011stochastic}, we now proceed to find the condition satisfied by the stationary covariance matrix $C^*$. Let $\Sigma=\mathcal{E}\mathcal{E}^\top$, then the Fokker-Planck equation governing the transition density $\emph{p}(X,t)$ is

\[
\displaystyle {\frac {\partial p(X ,t)}{\partial t}}=-\sum _{i=1}^{n}{\frac {\partial }{\partial x_{i}}}\left[(BX)_{ii} \, p(X ,t)\right]+\frac{1}{2}\sum _{i=1}^{n}\sum _{j=1}^{n}{\frac {\partial ^{2}}{\partial x_{i}\,\partial x_{j}}}\left[\Sigma_{ij} \,p(X ,t)\right],
\] or in a more compact form, provided $\Sigma$ is constant,

\begin{equation}\label{eqn:FokkerPlanck1}
\frac{\partial p}{\partial t}=-\nabla \cdot (BXp) +\frac{1}{2}\nabla \cdot (\Sigma \nabla p).
\end{equation}
Without loss of generality, suppose the initial condition is $X_0=0$, then the solution to \eqref{eqn:model1} is a \emph{Gaussian process} with mean $0$ and covariance  
\begin{equation}\label{eqn:cov1}
C_t\triangleq \int_{0}^{t}e^{B(t-s)}\mathcal{E}\mathcal{E}^\top e^{B^\top (t-s)}ds.
\end{equation}
Using Leibniz rule and differentiating \eqref{eqn:cov1} with respect to $t$ give 
\[
\frac{d C_t}{dt}=BC_t +C_tB^\top + \mathcal{E}\mathcal{E}^\top.
\]
Therefore, the asymptotic stationary covariance satisfies (setting $\frac{dC*}{dt}=0$) 
\begin{equation}\label{eqn: lyapunoveqn}
BC^* +C^*B^\top + \mathcal{E}\mathcal{E}^\top=0,
\end{equation}
which is known as a Lyapunov equation in Control theory.

Note that, because our model is a Gaussian process, the transition probability density for \eqref{eqn:FokkerPlanck1} is known at all times and is given by \cite{Gardiner:1985qr}
\begin{equation}\label{eqn:transitionpdf1}
p(X,t|X_0,0)=\frac{1}{(2\pi)^{n/2}\sqrt{det(C_t)}}\exp\left[ -\frac{1}{2}(X-e^{Bt}X_0)^T C_t^{-1} (X-e^{Bt}X_0)\right], 
\end{equation}
where the initial condition is  
\[
p(X,t|X_0,0)=\delta^n(X-X_0)\delta(t).
\]
%So far we are interested in weak solution to stochastic differential equations, that is the law of the process $X_t$. In many applications, the Brownian path is not known beforehand. The objective is to examine the law governing diffusion processes.  By dealing with the Fokker-Planck equation we can ask several questions:
%\begin{enumerate}
%	\item Does the transition density $p_t$ converge to a steady-state distribution $p_{st}$?
%	\item How fast is the convergence and can Wasserstein distance determine the rate of convergence?
%	\item What are the conditions on $A, Q$ and $\mathcal{E}$ that ensure $X_t$ has a regular probability density?
%\end{enumerate}
We now study the long-time stochastic opinion dynamics concerning the convergence and the property of the asymptotic (stationary) distribution.
%\subsection{Flux}
\section{Convergence to the Stationary Distribution}
\label{s: stationary distribution}
Recall that \eqref{eqn:model1} is a Gaussian process since the integrand in the It\^o integral is deterministic and bounded. At each point in time, a solution to the Fokker-Planck equation is also Gaussian, given by \eqref{eqn:transitionpdf1}. Does the long-time limit of $p_t$ reach a steady state? For $B$ having a negative real part of all the eigenvalues, many classic results show that $p_\infty$ exists since the dynamics is a multi-variate Ornstein-Uhlenbeck process. Here, however, we present another approach to investigate convergence to stationary distribution based on Wasserstein distances techniques in Optimal Transport theory, which has received increasing attention. Following \cite{lasota2013chaos}, we begin with the simplest case.

\subsection{Symmetric interaction}
Suppose that the discrete time interaction matrix $A$ is symmetric, then the infinitesimal generator matrix $Q$ is also symmetric. Thus, $B$ in \eqref{eqn:model1} is also symmetric. In this case, we can construct the Lyapunov function 
\[
V(X)=-\frac{1}{2}X^\top BX.
\]
As $B$ is symmetric, the gradient can easily be expressed as
\[
\nabla V=-BX,
\]
where the $\nabla$ is with respect to $X\in \R^n$. Therefore, the SDE reads
\[
dX_t=-\nabla V(X_t) dt+\mathcal{E}dW_t.
\]

In this case, we know that the stationary distribution will be reached eventually and is simply given by the Boltzmann distribution with energy $V$, and is discussed in more details in \ref{appendix b}. Now we will argue using Optimal Transport approach that the dynamics in fact converges to the Boltzmann distribution.
\subsection{Gradient flows}
For simplicity, we will assume $\mathcal{E}=\sqrt{2}\,\mathbf{I}$. Agents hold identical positive learning rates. The Fokker-Planck equation for the time evolution of the density of belief profiles of agents in $\R^n$ becomes

\begin{equation}\label{eqn:fokkerplanck3}
\frac{\partial p}{\partial t}=\nabla \cdot (\nabla V p + \nabla p ).
\end{equation}
Suppose $\mu_t$ is a solution  at time $t$ and $\nu$ is the stationary solution to \eqref{eqn:fokkerplanck3}, given some initial datum or starting point $X_0$. The stationary solution can be summarized, when the drift is a gradient of a convex funtion $V$, as

\[
\int e^{-V}=1.
\] 
and its stationary probability density is given by 
\[
d\nu(x)=e^{-V(x)}dx 
\] 

Then one can use techniques from functional analysis and use the Poincar\'e inequality to show convergence in the $L^2(\R^n;e^{-V})$ norm between probability measures weighed by the stationary measure $e^{-V}$.  A good reference for this is \cite{bolley2012convergence}. Another approach is to use Lyapunov function and examine stochastic stability \cite{khasminskii2011stochastic}. Both cases assume reversible diffusions. To generalize this and allow for non-gradient drift form a newer approach is required. This is a vast subject and we cannot do it justice here, especially the interplay between functional-analytic and Lyapunov methods. However, we outline a third approach, using Wasserstein distances, neatly summarized in \cite{bolley2012convergence}. We show in \ref{appendix c}, convergence to equilibrium using the Wasserstein approach. Recent work in this direction seems to show promise for the case of non-symmetric $B$ {\textemdash} when it cannot be expressed as the gradient of a convex functional.

\section{Synchronization of Agents: Criterion for Dimensionality Reduction of the Stationary Distribution}
\label{s: synchronization}

Till now our discussion has focused on linear SDEs with the same number of agents as the number of Brownian motions. Indeed, 
\eqref{eqn:model1} can be written in a more general setting

\begin{align} \label{eqn: SingleBrownian}
dX_t&=BX_t dt+\mathcal{E}dW_t,\\
X_0&=\xi,
\end{align}
where $W$ is an $r$-dimensional Brownian motion independent of the initial vector $\xi$, $B$ is an $n\times n$ matrix as before, and $\mathcal{E}$ is an $n\times r$ matrix. In general, systems could have $r<n$. 

Suppose the system has only one Brownian motion $(r=1)$ affecting all agents, then $\mathcal{E}$ becomes a column vector and $dW_t$ is \emph{one} dimensional. This reflects the condition where all the agent are subject to a common noise source in the proposed truth. The diffusion part of the SDE simplifies and can be written in two ways

\[
\mathcal{E}dW_t =\begin{pmatrix}
\varepsilon_1 \\ 
\vdots \\ 
\varepsilon_n
\end{pmatrix} dW_t^1 =\begin{pmatrix}
\varepsilon_1  &0& &0\\ 
\vdots &   & \ddots& \\ 
\varepsilon_n & 0 & &0
\end{pmatrix}\begin{pmatrix}
dW_t^1 \\ 
\vdots \\ 
dW_t^n
\end{pmatrix} .
\]

We will use the last representation, which means all agents are perturbed by only one Brownian motion (the noisy proposed truth $\bar \sigma$ is shared by every agent). For $n$ agents, the natural question to ask is if the transition density is absolutely continuous with respect to the Lebesgue measure? Since we are dealing with Gaussian measures at all times we can restrict our attention to covariance matrices.  With a single source of noise is the resulting covariance matrix singular or non-singular? We know from  \eqref{eqn:cov1} that the covariance matrix is
\[
C_t\triangleq \int_{0}^{t}e^{B(s)}\mathcal{E}\mathcal{E}^\top e^{B^\top (s)}ds.
\]
But now notice that here
\[
\mathcal{E}\mathcal{E}^\top=\begin{pmatrix}
\varepsilon_1 \\ 
\vdots \\ 
\varepsilon_n
\end{pmatrix} \begin{pmatrix}
\varepsilon_1,
\cdots,
\varepsilon_n
\end{pmatrix}=\begin{pmatrix}
\varepsilon_1  &0& &0\\ 
\varepsilon_2 &0   & &0 \\ 
\vdots &   & \ddots& \\ 
\varepsilon_n & 0 & &0
\end{pmatrix}\begin{pmatrix}
\varepsilon_1  &\varepsilon_2&\cdots &\varepsilon_n\\ 
0 &  0& &0 \\ 
\vdots &   & \ddots& \\ 
0 & 0 & &0
\end{pmatrix}.
\]

This is a rank one matrix. So the inner part in our definition of the covariance matrix is of rank one. This presents considerable difficulties. The process $X_t$ may be degenerate Gaussian.  However, in some cases even with $n$ agents it is not always the case that there is a manifold collapse {\textendash} the agents opinions lie in a lower dimensional space. It could be  that the covariance matrix is still nonsingular. To study this aspect we need to introduce the concept of controllability. In control theory, the typical issue is to determine an input which steers a dynamical system to a certain point at terminal time.

\begin{definition} 
	The pair of locally bounded functions  $(B,\mathcal{E})$ is controllable on $[0,T]$ if for every $x,y \in \R^n$, there exists a measurable, bounded function $v:[0,T] \to \R^n$ such that the time derivative
	\[
	\dot{Y} =  B(t)Y(t) + \mathcal{E}(t)v(t), \, 0 \leq t \leq T
	\]
	satisfies $Y(T)=y$: there exists a control function $v(\cdot)$ which steers the linear system from $Y(0)=x$ to $Y(T)=y$.  (See figure \ref{fig: controllability}).
\end{definition}

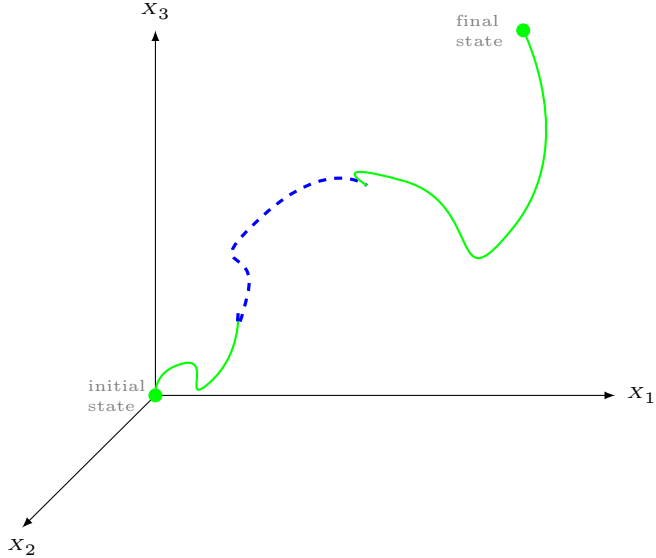
\begin{figure}
	\centering
	\begin{tikzpicture}[transform shape, scale = 1.21, thin,>=latex,font=\tiny,postaction={decorate,
		decoration={markings,mark=at position 0.5 with {\arrow{>}}}
	}]
	\draw[->] (0,0)--(0,4) node[above] {$X_3$};
	\draw[->] (0,0)--(-1.45,-1.45)node[below] {$X_2$};
	%\draw[->] (0,0)--(1,-1)	node[right] {$X_3$};
	\draw[->] (0,0)--(5,0)	node[right] {$X_1$};
	\filldraw [green] (0,0) circle (2pt) node[left,gray,text width=.61cm]{initial state} ;
	\filldraw [green] (4,4) circle (2pt)node[left,gray,text width=.61cm]{final state} ;
	\draw [green, thick] plot [smooth, tension=2] coordinates { (0,0) (0.3,0.35) (0.6,0.15) (0.9,0.9)};
	\draw [blue,dashed, very thick] plot [smooth, tension=2] coordinates { (0.9,0.9) (1.0,1.095) (1.16,1.96) (2.3,2.3)};
	\draw [green, thick] plot [smooth, tension=2] coordinates { (2.3,2.3) (2.7,2.35) (3.9,1.85) (4,4)};
	\end{tikzpicture}
	\caption{An illustration of controllability in three dimensions from an initial state to a desired state.}
	\label{fig: controllability}
\end{figure}

The non-degeneracy of the distribution of $X_t$, which is Gaussian, is precisely when the covariance matrix has full rank \cite{liu2011controllability}. This is explained in \cite{karatzas1998brownian} from which we state the definition and result. 

\begin{proposition}
	The pair of constant matrices $(B,\mathcal{E})$ is controllable on any interval $(0,T)$ if and only if the $n \times n$ controllability matrix
	\[
	[\mathcal{E}, B\mathcal{E}, B^2\mathcal{E},\cdots, B^{n-1}\mathcal{E} ]
	\] has rank $n$.
\end{proposition}

Actually, this result is proved in more generality by \cite{priola2008densities}, where the case of Ornstein-Uhlenbeck processes with jumps is also considered and leads us to the most important remark.

\begin{remark}
	Controllability is  equivalent to the covariance matrix having full rank. Thus, the probability law of $X_t$  is absolutely continuous with respect to the Lebesgue measure in $\R^n$.  If the dynamics are not controllable then the system collapses onto a subspace of $\R^n$. Consequently, the asymptotic covariance matrix will not be strictly positive definite and $\det(C^*)=0$.
\end{remark}

In linear systems theory, the covariance matrix is equivalent to a controllability grammian. If the grammian is positive definite, then it is nonsingular. To check whether our model has a probability density we only need to calculate the controllability matrix rank \cite{elliott1971consequence}. Consider the following examples where there is a single Brownian driver. 
\begin{figure}%{.5\textwidth}
	\centering
	\includegraphics[width=1 \linewidth]{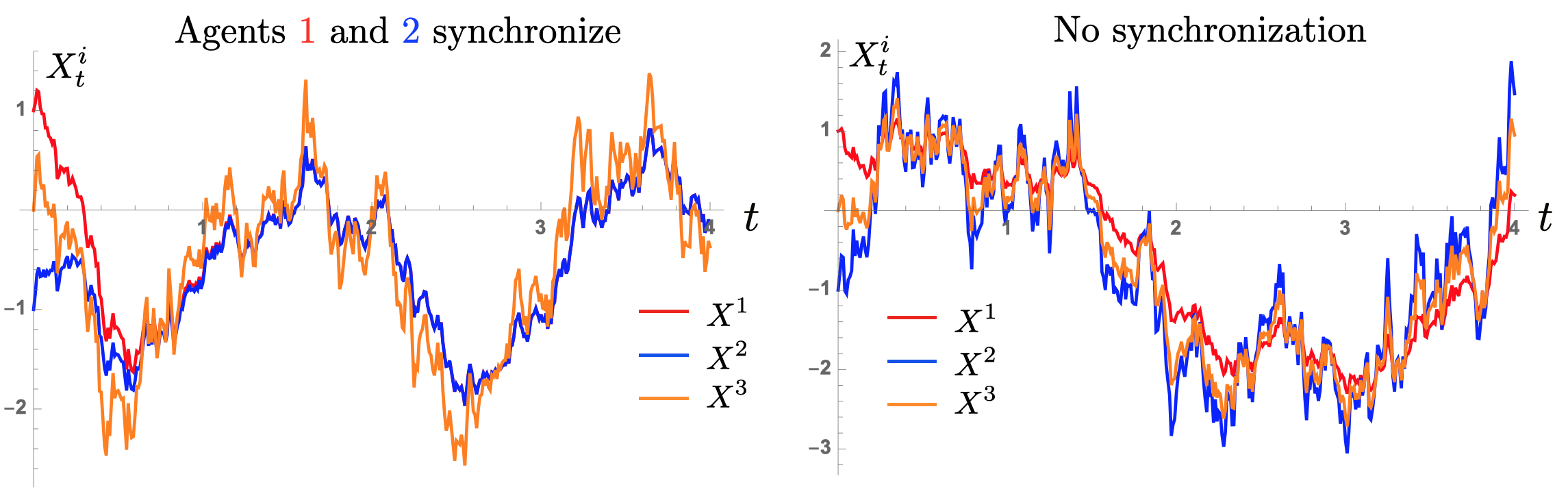}  
	\caption{Illustration of synchronization (dimensionality reduction) of noisy opinion dynamics of 3 agents. Initially, the agents' opinions are different at $(1,-1,0)$, respectively. (Left) However, at long times, noisy opinion dynamics from Example \ref{ex:notfullrank} in which the controllability matrix has rank 2 shows synchronization of agents 1 and 2. Stationary covariance matrix here has rank 2, and thus is degenerate, reflecting that correlated opinion dynamics at long times are 2 dimensional.  (Right) Noisy opinion dynamics from example \ref{ex:fullrank} in which the controllability matrix has full rank (3) does not show synchronization. The agents' opinions remain different. The covariance matrix here has full rank, and correlated opinion dynamics at long times lives in 3 dimensions.}
	\label{fig: synchronization}
\end{figure}%
\begin{example}\label{ex:notfullrank}
	As there is only one source of noise, the learning rates $\mathcal{E}$ are kept in column form.
	\[
	Q=\left(
	\begin{array}{ccc}
	-2  & 1 & 1 \\                
	1 & -2 & 1 \\
	1 & 2 & -3 \\
	\end{array}
	\right) ,\quad B=\left(
	\begin{array}{ccc}
	-3  & 1 & 1 \\                
	1 & -3 & 1 \\
	1 & 2 & -5 \\
	\end{array}
	\right) ,\quad\mathcal{E}=\left(
	\begin{array}{ccc}
	1  \\
	1\\
	2\\
	\end{array}
	\right).
	\]
	The rank of the controllability matrix
	$[\mathcal{E}, B\mathcal{E}, B^2\mathcal{E}]
	$ is 2, which is easily verified numerically by calculating the rank of matrix
	
	\[
	\begin{pmatrix}
	1  &   0  &  -7\\
	1  &   0  &  -7\\
	2 &   -7  &35
	\end{pmatrix}.
	\]
	Therefore, with a single Brownian motion, and two agents having the same learning rates we obtain a manifold collapse: the covariance matrix is singular.  In effect, with common noise agents 1 and 2 synchronize so that their opinions move together after some time period. Agent 3, however, is still learning at a different rate. The asymptotic covariance matrix is singular.
\end{example}
It seems in the preceding example, two agents having similar learning rates led to the covariance matrix being singular. On the other hand, let us consider another case with different learning rates.

\begin{example}\label{ex:fullrank}
	\[
	Q=\left(
	\begin{array}{ccc}
	-2  & 1 & 1 \\                
	1 & -2 & 1 \\
	1 & 2 & -3 \\
	\end{array}
	\right) ,\quad B=\left(
	\begin{array}{ccc}
	-3  & 1 & 1 \\                
	1 & -5 & 1 \\
	1 & 2 & -5 \\
	\end{array}
	\right) ,\quad\mathcal{E}=\left(
	\begin{array}{ccc}
	1  \\
	3\\
	2\\
	\end{array}
	\right).
	\]
The resulting controllability and asymptotic covariance (grammian) matrix are
	\[
	\begin{pmatrix}
	1	&2	&-17\\
	3	&-8	&33\\
	2	&-3	&1
	\end{pmatrix}, \quad
	C^*=\begin{pmatrix}
	0.99   & 1.39    &1.09\\
	1.39    &2.18    &1.67\\
	1.09    &1.67    &1.29
	\end{pmatrix}
	\]
	and have full rank. 
\end{example}

For a common noise, intuition may lead us to conclude that dynamics are also one dimensional. But this is not necessarily true. Although, if any of the agents learn at similar rates it is still possible that the phase space is a strict subset of $\R^n$, where $n$ is the number of agents. Generally, the grammian and controllability condition can help in determining non-singularity of the covariance matrix. However, the analytical condition for how and when the agents synchronize so that the phase space of beliefs live in a lower dimensional linear embedding of $\R^n$ remains an open question. 

%\begin{conclusion}
%	Controllability of a system implies that any point is reachable with control $u$ for system dynamics \ref{eqn: SingleBrownian}. Absolute continuity of a measure with respect to Lebesgue measure
%\end{conclusion}

\subsection{Lie Algebras and conditional density}
The controllability matrix can be thought of as a Lie bracket \cite{kashima2016noise}. The connection between Lie brackets and conditional densities received a great deal of attention in the early eighties, see \cite{brockett2014early} for a good survey. Lie brackets are used if $B$ and $\mathcal{E}$ are nonlinear functions \cite{whalen2015observability}. Though for this article, we will only examine linear drift and diffusion matrices. While our perspective is on consensus, control theory literature focuses on the network aspect of controllability: what types of networks are controllable and observable? These aspects are tangential to our main focus on agent behaviour. Control theory provides the techniques to steer a network of agents to a desired value or to control the system's outcome. Several engineering questions may need to be addressed in this case. For example, how many agents or nodes are needed for the system to be controllable. Another aspect would be to control the flow of information.

\section{Broken Detailed Balance and Non-equilibrium Steady States}
\label{s: ness}

%Let us illustrate the meaning in 2 agent dynamics. m
%Thip to put graphs and show synchronization graphs
%B=[-1.5 2, 3 -2.8]  $\mathcal{E}$=[1/2 0,0 2/10]
%sigmabar or X_E=2

Now we report the broken-detailed balance phenomenon arising in our continuous-time opinion dynamics model with noisy feedback. We will demonstrate that even in a simple two-agent learning dynamics, where both agents are susceptible to a common noise source, there is an asymptotic stationary distribution with broken detailed balance. Thus, this simple opinion dynamics gives rise to a non-equilibrium steady state with a non-zero asymptotic probabilistic current loop discussed in Refs.\cite{Mellor:2016aa, Mellor:2017aa, Zia:2007he}. 

For two learning agents, the stochastic dynamics with a common noise and with the proposed truth $\bar \sigma = 0$ follows from Eq. (\ref{eqn: SDE_WhiteNoise}):
\begin{align}\label{eqn: 2agent_dyn}
d X^1 &=  Q_{12}(X^2 - X^1) - \varepsilon_1 X^1 + \varepsilon_1 dW_t\\
d X^2 &=  Q_{21}(X^1 - X^2) - \varepsilon_2 X^2 + \varepsilon_2dW_t, \nonumber
\end{align}
where $dW_t$ is a \emph{one dimensional} Brownian motion.

Following the notation in \autoref{s: noisy}, this can be rewritten as \begin{equation}\label{eqn: 2d_dynamics}
dX = BX + \Delta \mathcal{E}  {\bf \Gamma}_t dt,
\end{equation}
where ${\bf \Gamma}_t =   \left( \begin{smallmatrix}\Gamma_t \\ \Gamma_t\end{smallmatrix}\right)$ is a Gaussian white noise uncorrelated in time, $\Delta$ is a scaling of the Brownian noise, $B = Q - \mathcal{E}$, with $Q_{11} = - Q_{12},$ $Q_{22} = - Q_{21}$, and  $\mathcal{E} = \text{diag}(\varepsilon_1, \varepsilon_2) = \left( \begin{smallmatrix}\varepsilon_1 & 0 \\ 0 & \varepsilon_2 \end{smallmatrix}\right).$ 

Unlike usual statistical mechanics problems where the noise in each coordinates are typically assumed uncorrelated, the noise in each coordinate here is {\it identical}, encapsulating the situation that the two learning agents are influenced by a common source of noisy information. Nevertheless, we know the asymptotic distribution generated from the above multivariate Ornstein-Uhlenbeck process is a Gaussian distribution with the covariance satisfying Lyapunov equation \eqref{eqn: lyapunoveqn}, see also the derivations in \cite{Risken:1984df,Van-Kampen:1992tk,Gardiner:1985qr} :
\begin{equation}\label{eqn: Lyaupunov}
BC^* + C^* B^T  + D = 0,
\end{equation}
where $D$ is the noise correlation matrix whose matrix elements are defined by
\begin{align*}
D_{ij}\delta(t-t') &= 
\Delta^2\left(\big\langle \mathcal{E}{\bf \Gamma}_t{\bf \Gamma}^T_{t'}\mathcal{E}^T\big\rangle\right)_{ij} \\
&= \Delta^2\varepsilon_i \varepsilon_j \langle\Gamma_t\Gamma_{t'}\rangle \\
&= \Delta^2\varepsilon_i \varepsilon_j \delta(t-t'). 
\end{align*}
For standard Brownian motion and in the subsequent analysis we set $\Delta=1$,
so we identify  (as in \eqref{eqn:cov1})
\begin{equation}\label{eqn: Diffustion Matrix}
D_{ij} = \varepsilon_i \varepsilon_j\Delta^2=\varepsilon_i \varepsilon_j.	
\end{equation}

For two dimensional dynamics, the stationary covariance 
$C^*$
%$C^*=\lim_{t \rightarrow \infty} \E [X_t\,X_t^\top]$
%$C^* \equiv \lim_{t \rightarrow \infty} \langle x_i(t) x_j(t) \rangle$ 
is exactly solvable, see  Ref. \cite{Gardiner:1985qr}, and given by 
\begin{equation}\label{eqn: 2dcovariance}
C^* \equiv \frac{\left(\text{det}B\right)D + \left[B - \left(\text{tr}B\right)I\right]D\left[ B - \left(\text{tr}B\right)I\right]^T}{2\left(\text{tr}B\right)\left(\text{det}B\right)},
\end{equation}
where $I$ is the $2 \times 2$ identity matrix.

In our case, $\text{tr}B =- \left( \varepsilon_1 + \varepsilon_2 + Q_{12} + Q_{21}\right)$ and $\text{det}B = \left( \varepsilon_1 \varepsilon_2 + \varepsilon_2 Q_{12} + \varepsilon_1 Q_{21} \right)$. Straightforward calculations give
\begin{equation}\label{eqn: covariance_st}
C^* =  \frac{1}{2\text{tr}B}\left( 
\begin{matrix} 
\text{det}B + \varepsilon_1^2 & 
\text{det}B + \varepsilon_1\varepsilon_2 \\ 
\text{det}B + \varepsilon_1\varepsilon_2&
\text{det}B + \varepsilon_2^2
\end{matrix}
\right),	
\end{equation}
whose eigenvectors and the corresponding eigenvalues are, respectively, 
\begin{align}
V_{\pm} &=
\left( 
\begin{matrix} \frac{\veps_1^2-\veps_2^2 \ \pm \ \sqrt{4\left(\text{det}B + \veps_1\veps_2\right)^2 + \left(\veps_1^2-\veps_2^2 \right)^2}}{2(\text{det}B + \veps_1\veps_2)} \\ 1 \end{matrix}
\right), \\ 
\lambda_{\pm} &= \frac{2\text{det}B + \veps_1^2 + \veps_2^2 \ \pm \ \sqrt{4\left(\text{det}B + \veps_1\veps_2\right)^2 + \left(\veps_1^2-\veps_2^2 \right)^2}}{4\text{tr}B}.
\end{align}

The correlation matrix, its eigenvalues, and eigenvectors determine the profile of the stationary Gaussian distribution

\begin{equation}
p_{st}(X) = \frac{1}{2\pi\sqrt{ \text{det}C^* }}\exp \left[ -\frac{1}{2}X^T C^{*-1} X\right].
\end{equation}

It is interesting to note that when both agents learn at identical rate $\varepsilon_1 = \varepsilon_2 = \varepsilon$, the stationary covariance matrix $C^*$ becomes degenerate with the eigenvectors $\left( \begin{smallmatrix}-1 \\ 1 \end{smallmatrix}\right)$ and $\left( \begin{smallmatrix}1 \\ 1\end{smallmatrix}\right)$ corresponding to the eigenvalue $0$ and $\varepsilon$, respectively. This is because the noises cancel out in the $(X^1 - X^2)$ coordinate and Eq.(\ref{eqn: 2d_dynamics}) gives
$\frac{d}{dt} (X^1 - X^2) = - \left(Q_{12}+Q_{21} + \varepsilon \right)(X^1 - X^2),$
meaning that both agents will eventually synchronize at $X^1 = X^2$ since $\left(Q_{12}+Q_{21} + \varepsilon \right) > 0.$ Consequently, when $\veps_2 = \veps_2 = \veps$, the stationary distribution collapses onto a one-dimensional manifold defined by the line $X^1=X^2$, with the variance given by the eigenvalue $\veps$ corresponding to the eigenvector $\left( \begin{smallmatrix}1 \\ 1\end{smallmatrix}\right)$. 

On the other hand, when $\veps_1 \neq \veps_2$, the stationary distribution is no longer constrained on a one-dimensional synchronization manifold $X^1=X^2$; in addition, it is also a non-equilibrium steady state (NESS)\cite{Zia:2007he}. To see this, we first rewrite the Fokker-Planck equation as
\begin{equation}
	\frac{\partial p}{\partial t}=-\nabla \cdot (BXp -\frac{1}{2} \nabla \cdot (\Sigma p) ) = -\nabla \cdot {\bf J}\label{eqn:fokkerplanck2},
\end{equation}
where we identify the probabilistic flux (current) as
\begin{equation}\label{eqn:flux}
	{\bf J}(X,t) = (BX)p(X,t) -\frac{1}{2} \nabla \cdot \left[\Sigma p(X,t)\right].
\end{equation}
In the coordinate form, the stationary state flux is
\begin{equation}\label{eqn: prob_current}
J_{st}^i(X) = \left[ \left( BX  \ \right)^i - \frac{1}{2}\sum_j\left( D_{ij} \partial_{X^j}\right) \right] p_{st}(X).
\end{equation}
By conservation of probability, $0 = \lim_{t \rightarrow \infty} \partial_tp(X,t) = \partial_t p_{st}(X) = \nabla \cdot {\bf J}_{st}(X). $
Thus, at the stationary state, either ${\bf J}_{st}(X) = {\bf 0}$ or ${\bf J}_{st}(X) = \nabla \times {\bf f}$ for some non-zero flow field ${\bf f}$. The former reflects the detailed balance condition, in which the net probabilistic flux flowing in and out of a volume element surrounding any point $X$ is zero. The latter corresponds to a broken detailed balance, in which non-zero probabilistic current {\it loop} persists in the steady states. This current loop is a feature of a NESS \cite{Zia:2007he} which naturally appears in diverse non-equilibrium systems \cite{Zia:2007he, Mellor:2016aa, Mellor:2017aa, Gnesotto_2018, PhysRevE.97.052121, zhang2017landscape}. Under what conditions NESS exists and is stable is an interesting and complementary feature that is under active research  recently \cite{liverpool2020steady}.

\begin{figure}%{.5\textwidth}
	\centering
	\includegraphics[width=1 \linewidth]{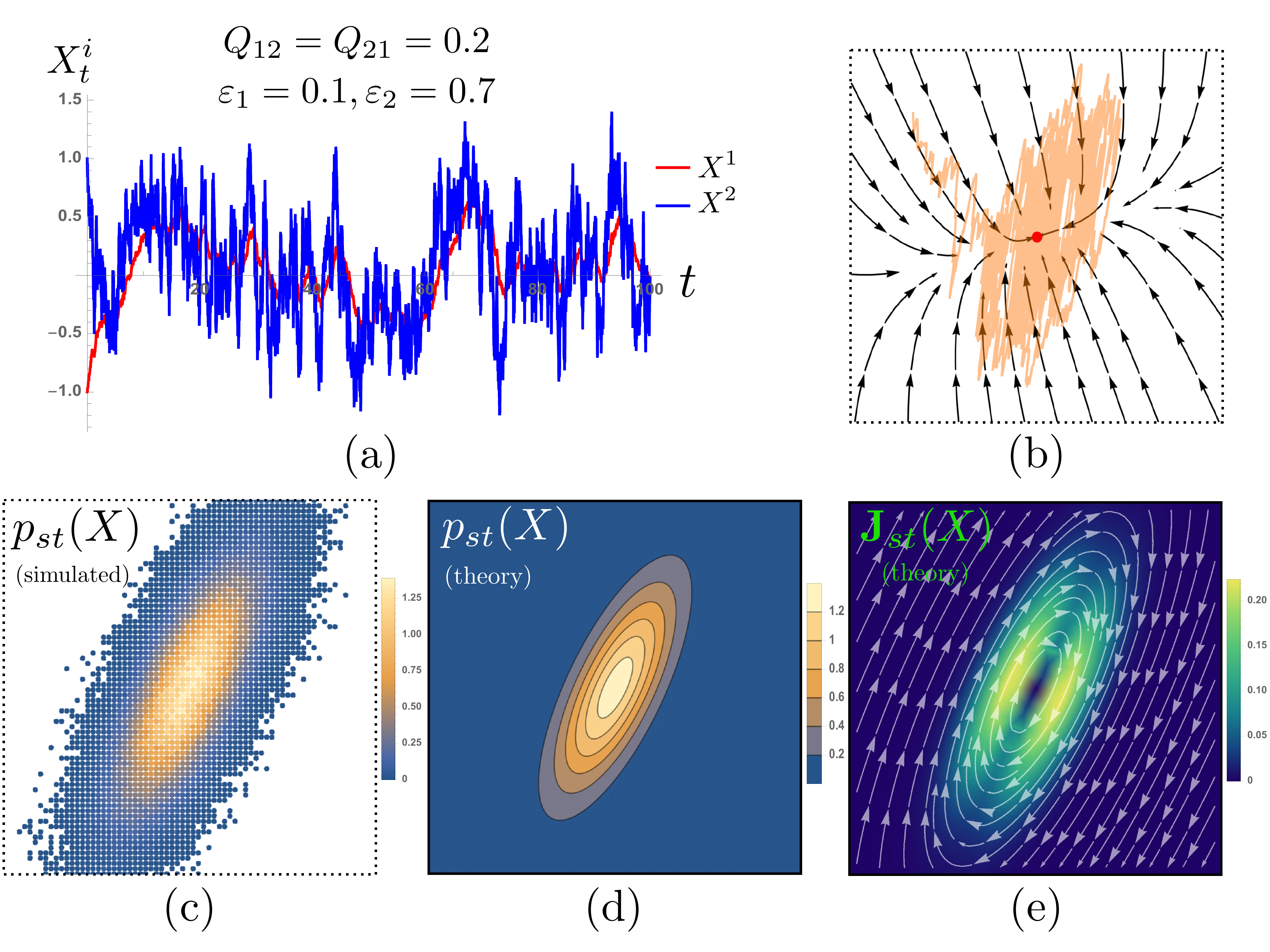}  
	\caption{Non-equilibrium steady state (NESS) associated with 2-agent opinion dynamics under a common noise with $Q_{12} = Q_{21} = 0.2, \veps_1 = 0.1, \veps_2 =0.7, \Delta = 1.$ As opposed to when $\veps_1 = \veps_2$, in which the long-time dynamics collapse onto a one-dimensional synchronization manifold $X^1 = X^2$, the dynamics here is fully two-dimensional with highly correlated agents' opinions. (a) show correlated dynamics generated from Eq. (\ref{eqn: 2agent_dyn}) with the initial condition $X_0 = (-1,1)$, using stochastic Runge-Kutta algorithm, while (b) visualizes the trajectories (orange) in $X^1,X^2$ phase space. Black arrows are the vector field associated with deterministic dynamics $d X/dt = BX$ of Eq. (\ref{eqn: 2d_dynamics}). For the parameters considered here, the fixed point (red) at $(0,0)$ is stable, with the linear stability matrix $B$ having all real negative eigenvalues without imaginary parts, i.e. $\sqrt{\text{tr}^2B - 4\text{det}B} = 0.52 >0$. However, the interplay between deterministic relaxation and flucutations due to a common noise leads to a NESS with a non-zero probabilistic current loop (e). The stationary density (histogram) shown in (c), constructed from sampling $2\times 10^3$ independent realizations of the steady state, agrees well with the predicted Gaussian distribution with the covariance $C^*$ given by (\ref{eqn: covariance_st}), which encodes the correlation $\E [X^1_sX^2_t].$ The eigendirection with the largest eigenvalue of $C^*$ rotates away from the synchronization manifold $X^1=X^2$ toward the axis of a faster learning agent (towards $X^2$-axis), reflecting that a faster learner is more susceptible to a common noise.  Lastly, the stationary flux ${\bf J}_{st}(X)$ of Eq.(\ref{eqn: prob_current}) encodes two-time correlation function $\E [X^i_s X^j_t]$. The vector (flux) field (white arrow) exhibits a non-zero current loop, reflecting a broken detailed balance. The colours encode the strength of the flux. Figures (b)-(e) are shown on the axes' scale $[-1.5,1.5]\times [-1.5,1.5]$.}
	\label{fig: NESS}
\end{figure}%

As a consequence of  broken detailed balance, the NESS in our two-agent model exhibits correlated oscillation of agents' opinions, despite the imaginary part of the eigenvalues of $B$ being zero, i.e. $ \sqrt{\text{tr}^2B - 4\text{det}B} \in {\rm I\!R}$. This differs from the well-known mechanism of noise-induced oscillation that arises when the imaginary parts of the linear stability matrix are non-zero, such as in the predator-prey model \cite{McKane2005}. See Fig. \ref{fig: NESS} for a more detailed discussion.
\clearpage
\section{Discussion and Outlook}
Studying SDEs with a common noise presents considerable difficulties. To ensure that the process still possesses a regular density, we have to calculate a controllability matrix and be aware of subtle mathematical issues. In mathematics papers, the object of study is the SDE itself and what conditions ensure regularity of the transition and asymptotic densities. However, in many of these papers no specific example is given. This paper provides interesting examples within the context of a modelling framework where such regularity conditions are needed. In many instances, theoretical papers fail to provide a strong physical interpretation. By recasting our problem as social learning, we are able to link with theoretical aspects and provide simple yet insightful examples. The connection with the work of \cite{Zia:2007he, Mellor:2016aa} is even stronger; the dynamic social learning models developed here provide specific instances of non-equilibrium steady states (NESS). Consequently, our work demonstrates a class of simple social learning models that possesses a NESS.  A potential future source of work could be to fractional brownian motion \cite{eab2018ornstein, zhang2014stochastic, ren2007asymptotic, li2019stochastic, gajda2014fokker}. Though the interpretation for financial markets seems unclear and would have to be modelled correctly in the context of opinion dynamics. A natural extension would be to consider costs and time-varying network topologies of trust matrix $A$ \cite{MooreNewmanPRE2004, zhang2017random, weber2019deterministic}. Indeed, this would open up connections with mean-field limits  of interacting agents \cite{ieda2011modeling, jabin2017mean}. Naturally, the question then becomes of what type of interaction structure is suitable as the number of agents tend to infinity.

One implication from our analysis is clear. Being a faster learner with higher learning rate $\varepsilon$ translates into more sensitivity to noise.  Is this a good property? In financial markets, for example, where agents are all learning and revising their quotes for a stock, the faster learner bears a cost in adjusting her opinion too frequently. At times the noise doesn't convey any real information.  At least that is the case when all agents are hit by a common source of noise. From our analysis of the Kalman controllability matrix, it is evident that even if some agents use the same learning rates, there could be a collapse of the asymptotic density to a lower dimensional subspace of $\R^n$. Due to the nature of interaction, as long as agents are minimally interacting, a few can be stubborn and not learn with $\varepsilon=0$. This, however, doesn't impact the whole system dynamics converging to a stationary measure. When there are more sources of uncertainty, it is not clear whether being a faster or slower learner is better.  Thus far, we have not considered costs that agents bear when they update and revise. \ref{appendix a} gives detail on possible ways a cost function may be constructed. The rapid electronification of markets where interaction is highly visible presents new challenges to existing models.

\section*{Acknowledgements}

Georgios Piliouras acknowledges MOE AcRF Tier 2 Grant 2016-T2-1-170,  grant PIE-SGP-AI-2018-01, grant NRF2019-NRF-ANR095 ALIAS and NRF 2018 Fellowship NRF-NRFF2018-07. Tushar Vaidya acknowledges the SUTD President's Graduate Fellowship. Thiparat Chotibut would like to thank Shaowei Lin for helpful discussions, and acknowledge the financial support from SUTD-ZJU grant
ZJURP1600103 and Chulalongkorn University's new faculty start-up grant. 

\bibliography{mybibfile}

\appendix
\section{Game Theoretic Aspects of Degroot Learning}\label{appendix a}
Following Refs. \cite{BINDEL2015248, ghaderi2014opinion}, we discuss the connection of the DeGroot model with and without experts (stubborn agents) to game theoretic concepts in economics and computer science.  Nash equilibria is a key concept in economics but increasingly in computer science as it has implications for equilibrium and computational complexity.  %Certain games may not even have an equilibrium. 
Here we give a heuristic argument of how our update rule (\ref{eqn:model1}) may arise with a given cost function. 

Opinions are continuous between $[0,1]$ and each agent has an initial opinion. If we have a cost function
	\[
	C_i(X^i)= \sum_{j=1}^nA_{ij}\left( X^j-X^i \right)^2 - \varepsilon_i \left( \ms - X^i \right)^2
	\]
then it is optimized when the first order derivative equals zero with respect to $X^i$. The neighbours of agent $i$ or those with positive weight: $A_{ij}>0$. The first order conditions from above give us the following updating rule (best response strategy):

\[
X^i_{t+1}= \sum_{j}^n A_{ij}X^j_t+\varepsilon_i(\ms - X^i_t). 
\]

Precisely, this is our original update rule in discrete time. We assume that agents have self-belief $A_{ii}>0$ and learning rates are positive $ \varepsilon_i \in (0,A_{ii}) $, the same conditions as in \cite{vaidya2018learning}. Our Degroot updaters are myopic; they only care for the one period ahead forecast and cost. If all agents are following this rule, which is the case in our model, then it is a coordination game and everybody reaches the unique Nash equilibrium.

\[
\lim_{t\rightarrow\infty} X^1_t=X^2_t=\cdots=X^n_t.
\]

The agents don't have an incentive to deviate from their best response. Speed of convergence and stubborn agents are discussed at length in \cite{ghaderi2014opinion} $[Lemma \, 1]$, where consensus is defined to be a convex combination of initial opinions of stubborn players. Actually our model reduces to the model of stubborn agents. If we only have one stubborn agent $X^S$, then everybody converges to this opinion; i.e., $\lim_{t \rightarrow \infty}X^i_t = X^S$. While convergence time is interesting, our focus is on the connections to game theoretic ideas. The fact we have noise complicates the picture and the concept of Nash equilibrium in its usual static case needs to be reinterpreted in terms of the mean of the stationary distributions.

Averaging opinions in social networks can be seen as optimal in myopic sense. Agents optimize their own costs. If a social planner was to optimise the total costs then perhaps another strategy and optimal solution may exist. In this case, the \textit{price of anarchy} captures how selfish agents may guide the system to one equilibrium that may not be efficient from society's point of view. In actuality, for the dynamics we have taken in the main body of paper, the agents are myopic but in financial markets there is no central planner. Prices of assets are determined by interaction and complex dynamics. The very essence of free markets means we take optimality to be the price or consensus value determined by interacting players.

\section{Existence of Solutions and the Boltzmann Distribution}\label{appendix b}
The Fokker-Planck equation is the derivative of the transition density with respect to time and is otherwise known as the forward Kolmogorov equation \cite{thomas2019phase}.  Since the dynamics are specified by an SDE, which is a Gaussian process in our case, the transition density is Gaussian and so is the asymptotic density $p_{st}$. However, we cannot solve the long-term covariance matrix $C^*$ from Lyapunov equation \eqref{eqn: lyapunoveqn} analytically. One way to examine the long-term density is to investigate the PDE.  Factoring out the grad $\nabla$ operator, the Fokker-Planck equation \eqref{eqn:FokkerPlanck1} reads
\begin{equation}
	\frac{\partial p}{\partial t}=\nabla \cdot (-BXp +\frac{1}{2} \nabla \cdot (\Sigma p) );
\end{equation}
with the initial condition
\begin{equation}
p(X,0)=p_0(X),	
\end{equation}
and the finite support condition
\begin{equation}
\lim_{|X|\rightarrow\infty}p(X,t)=0.
\end{equation}
There is a decaying condition at infinity, there is no boundary as such. It ensures that the solution to the Fokker-Planck equations remain Gaussian. Let us assume the initial condition is  distributed according to $p_0(X)$.  A solution to the Fokker-Planck equation provided it exists is a function in $C^{2,1}(\R^n,\R^+)$, twice differentiable in $x$ and once in $t$. The equations above seem to suggest a divergence form. Let us abbreviate the drift which is $BX\in \R^n$ as $\textbf b(X)\in \R^n$. Furthermore, as $B$ is constant its constants can be embedded in a new drift function $\textbf b:\R^n \to \R^n$ and the SDE model is recast as
\[
dX_t=\textbf b(X_t) dt+\mathcal{E}dW_t.
\]
The diffusion matrix 
\begin{equation}
\Sigma=\mathcal{E}\mathcal{E}^T
\end{equation}
is symmetric and non-negative. In our case it is a diagonal as the learning rates matrix $\mathcal{E}$ is strictly positive. To ensure that a solution to the initial value Fokker-Planck equation exists for all times $t >0$ and is unique, we need a condition \cite{evans2012introduction}. 
\begin{assumption}[Uniform Ellipticity]
	The diffusion matrix is uniformly positive and there exists a constant  $\theta >0$, such that
	\[
	Y^\top \Sigma(X) Y \geq \theta \norm{Y}^2, \quad \forall Y \in \R^n
	\]
	uniformly for $X \in \R^n$.
\end{assumption}

This condition is necessary to ensure the existence of a transition probability density. Moreover, the assumption is equivalent to $\Sigma$ being non-singular. For the case of the diffusion matrix $\Sigma$ being constant, the uniform condition is trivially satisfied. However, earlier we assumed that the learning rates have to be strictly positive. The ellipticity assumption  is conveying something a bit stronger. Learning rates can't be arbitrarily small.  In a latter section, we will discuss the issue of a transition density still existing even when $\Sigma$ is singular and introduce the notion of controllability. 

\begin{proposition}
	Uniform ellipticity implies learning rates $\mathcal{E}>\theta$, for some constant $\theta>0$.
\end{proposition}

So in fact, if $\mathcal{E}$ not greater than zero, then some agents are not learning and a solution to  \autoref{eqn:FokkerPlanck1} may not exist. Further conditions are needed for the solution to be unique and regular with respect to the Lebesgue measure in $\R^n$. When the learning matrix is an identity matrix, $\Sigma=\mathbbm{I}$,  uniform ellipticity is satisfied for some $\theta \in (0,1)$. To see this, observe

\[
Y^\top\Sigma Y=Y^\top \mathbbm{I}Y=\norm{Y}^2>\theta \norm{Y}^2.
\]

%expand here

%example
\begin{definition}\label{eqn: fluxdef}
	The probability flux (current) is the vector
	
	\[
	\textbf{J}(X,t) \coloneqq\textbf b(X)p -\frac{1}{2}\nabla \cdot (\Sigma(X)p).
	\]
\end{definition}
This flux is seen in the earlier version as \eqref{eqn:fokkerplanck2}, which allows one to write the Fokker-Planck equation in the conservation law form as 
\begin{equation}
\frac{\partial p}{\partial t}+\nabla \cdot \textbf J=0.
\end{equation}
At the stationary state $\frac{\partial p_{st}}{\partial t}=0$, which means either ${\bf J}_{st}$ the stationary flux is zero and detailed balance is preserved or ${\bf J}_{st}=\nabla \times \textbf f \neq 0$ with a broken detailed-balance; an example of this in a two-player scenario provided in \autoref{s: ness}.

Suppose we have constant diffusion matrix $\mathcal{E}$ and hence constant $\Sigma$. The running assumption of the diffusion matrix being uniformly elliptic means the inverse $\Sigma^{-1}$ exists. At the steady state, suppose the detailed balance is preserved then the stationary flux must vanish

\[
\textbf b(X)p_{st} -\nabla \cdot (\Sigma p_{st})=0.
\]
Then as $\Sigma$ is constant, the above equation becomes, after distributing the grad operator over $\Sigma p_{st}$ and noting that $\nabla \cdot \Sigma=0$,
\begin{align*}
\textbf b(X)p_{st} -\nabla \cdot (\Sigma p_{st})&=0\\
\textbf b(X)p_{st} -(\Sigma \nabla p_{st} +( \nabla \cdot \Sigma)p_{st})&=0\\
\textbf b(X)p_{st}-\Sigma \nabla p_{st} &=0.
\end{align*}
Informally, the last equation reveals $\frac{\nabla p_{st}}{p_{st}}=\textbf b\Sigma^{-1}$. For the flux to vanish $V=\textbf b\Sigma^{-1}$ is the gradient of a function as the left hand side is; since the curl of a gradient vanishes. And since the right hand side is a gradient, its curl is also zero, implying that
\[
\frac{\partial V^i}{\partial X^j }=\frac{\partial V^j}{\partial X^i }.
\]
Integrating both sides and recalling that $p_{st}$ is a density, so always greater than zero,
\begin{align*}
\int \frac{\nabla p_{st}}{p_{st}}&=\int \textbf b\Sigma^{-1}\\
\log p_{st} &= \int \textbf b\Sigma^{-1},
\end{align*}
which gives the Boltzmann-like distribution for the steady state distribution
\begin{equation}
	p_{st}=\frac{1}{Z}e^{-V},
\end{equation}
where $V= \int \textbf b\Sigma^{-1}=\frac{1}{\Sigma} \int \textbf b$ is treated as a potential function (Hamiltonian) and there is a normalization constant that resembles a partition function $Z=\int_{\R^n} e^{-V}$. The argument is heuristic and we assume all the necessary nice conditions hold for the integration to make sense. 
\section{Convergence to the Stationary Distribution: Optimal Transport Approach}\label{appendix c}
A solution to \eqref{eqn:fokkerplanck3} can be thought of a mapping a path of $\mu_t$ in the space of probability measures $\mathcal{P}_2(\R^n)$ converging to the large-time limit or asymptotic distribution.
Here $\mathcal{P}_2(\R^n)$ plays the role of a space containing probability measures of the diffusion processes with finite second moments. The metric on this space is the Wasserstein distance between measures. We take \cite{bolley2012convergence}'s cue. Recall the definition of Wasserstein distance between two measures $\rho_1, \rho_2$
\begin{definition}
	
	\[
	W_2(\rho_1,\rho_2)=\inf_{\pi \in \Pi(\rho_1, \rho_2)}(\int_{\R^n}\int_{\R^n} |x -y|^2d\pi(dx,dy))^{1/2}
	\]
	with joint distribution $\pi$ having marginals $X\sim\rho_1$ and $Y\sim \rho_2$, and  $\Pi$ the space of all such joint measures. 
\end{definition}
Wasserstein distances metrize weak convergence (distribution convergence) and so it is natural to study the flow to equilibrium in the space of $\mathcal{P}_2(\R^n)$ instead of the process $X_t$. If we write \eqref{eqn:fokkerplanck3} in the SDE form, it reads
\begin{equation}
dX_t =-\nabla V dt+\sqrt{2}dW_t \label{eqn:smoluchowski1}.	
\end{equation}
 Equation \eqref{eqn:smoluchowski1} is an example of an Ornstein-Uhlenbeck process with gradient structure.  Suppose that there are two initial random starting points $X_0$ and $Y_0$ with probability laws $\mu_0,\nu_0$. Furthermore, let $X_t$ and $Y_t$ be solutions to  \ref{eqn:smoluchowski1} with the same Brownian driver.  With

\[
X_t=X_0 e^{Bt} +\int_{0}^{t}e^{B(t-s)}\mathcal{E}dW_s,
\]

$X_t-Y_t$ becomes simply $X_t-Y_t=X_0 e^{Bt}-Y_0 e^{Bt}$ because the diffusions cancel out. Similar to a coupling argument in discrete time Markov chains, our analysis couples the processes $X_t$ and $Y_t$. This type of coupling argument is standard for both discrete and continuous time Markov processes. Notice that time derivative for the squared distance simplifies
\[
\frac{d}{dt}|X_t-Y_t|^2=2(X_t-Y_t)\cdot(BX_t -BY_t)=-2(X_t-Y_t)\cdot(\nabla V(X_t)-\nabla V(Y_t)).
\]
Assume that Lyapunov function is strongly convex.  Denote $\tilde{B}=-B$ and the Lyapunov function as

\[
V=-\frac{1}{2} x^{\top} B x=\frac{1}{2} x^{\top} \tilde{B} x.
\]

Strong convexity is needed here and key in gradient flow type of arguments. Here $\nabla V= \tilde{B}x$ and $\nabla^2 V=\tilde{B}$ is positive definite
\[
\tilde{B} \succeq C\, \mathbbm{1}
\]
for positive real constant $C$ and identity matrix $\mathbf{I}$. Consequently $(\tilde{B}-C\, \mathbf{I})\succeq 0$ from which it follows
\begin{align*}
z^{\top} \tilde{B} z &\geq z^\top C\, \mathbbm{1} z, \quad \forall z \in \R^n \implies\\
z^{\top} \tilde{B} z &\geq C |z|^2 .
\end{align*}
Strong convexity is sometimes referred to as uniform ellipticity. The two conditions are the same. Substitute $(x-y)$ for $z$ and simplify
\begin{align}
(x-y)^{\top} \tilde{B} (x-y)&\geq C |x-y|^2 \nonumber\\
(x-y)^{\top} (\tilde{B}x-\tilde{B}y)&\geq C |x-y|^2\nonumber\\
(x-y)\cdot (\nabla V(x)-\nabla V(y))&\geq C |x-y|^2 \label{eqn:stroncvx}.
\end{align}
Multiplying condition \ref{eqn:stroncvx} by $-2$ we finally get that the time derivative of the coupled squared distance process is

\begin{equation}\label{eqn:coupledSDEdistance}
\frac{d}{dt}|X_t-Y_t|^2=-2(X_t-Y_t)\cdot(\nabla V(X_t)-\nabla V(Y_t)) \leq -2C|x-y|^2 \quad \forall x,y \in \R^n.
\end{equation}
This is a simple ODE with starting point $(X_0-Y_0)^2$ so integrating with respect to time and taking expectations yields
\begin{align*}
|X_t-Y_t|^2 &\leq e^{-Ct}|X_0-Y_0|^2\\
\E |X_t-Y_t|^2 &\leq e^{-Ct}\E |X_0-Y_0|^2.
\end{align*}
By definition the Wasserstein distance is the infimum of $\E |X_t-Y_t|^2$ over joint measures satisfying the marginal laws $X_t\sim \mu_t$ and $Y_t \sim \nu_t$. Naturally this implies
\[
W_2^2(\nu_t,\mu_t) \leq \E |X_t-Y_t|^2
\]
and more importantly 
\begin{equation}\label{eqn:Wasscontraction1}
W_2(\nu_t,\mu_t) \leq e^{-Ct} W_2(\mu_0,\nu_0).
\end{equation}
Contraction \ref*{eqn:Wasscontraction1} is vital to prove that in large time the density converges.  So far we have shown that we have a contraction mapping in the space of probability measures with metric $W$.  Villani \cite{villani2009optimal}[Theorem 6.18] shows that the Wasserstein space over a Polish space itself is a Polish space: complete separable metric space. Thus Cauchy sequences converge to a limit in the space of measures.

The interaction matrix being symmetric and $V$ being strongly convex are in some sense restrictive assumptions when modelling coordination games or social dynamics; nevertheless they form the starting point for sophisticated analysis.  For the Smoluchowski model \ref{eqn:smoluchowski1}, we know the stationary distribution $\nu$. Thus the Wasserstein contraction, if we choose $\nu_0$ as the stationary distribution, is
\[
W_2(\mu_t,e^{-V}) \leq e^{-Ct} W_2(\mu_0,e^{-V})  \quad \forall t \geq 0.
\]
Thus $e^{-V}$ is the only stationary distribution. So far the discussion was on drifts that can be expressed in gradient form. For the less restrictive case a more careful examination is needed.
\subsection{Non-symmetric interaction}
When the drift matrix $B$ is allowed to be non symmetric then the techniques to show convergence to a Gaussian distribution are more involved.  However, since matrix $B$ and $\mathcal{E}$ are constant it can be shown that the resulting asymptotic density is unique.

\begin{proposition}\cite{lasota2013chaos}
	Assume $B$ is asymptotically stable and $\mathcal{E}$ is uniformly elliptic.  The limiting function $p_t$ of \ref{eqn:fokkerplanck2} satisfies
	
	\[
	-\sum _{i=1}^{n}{\frac {\partial }{\partial x_{i}}}\left[(BX)_{ii} \, p(X ,t)\right]+\frac{1}{2}\sum _{i=1}^{n}\sum _{j=1}^{n}{\frac {\partial ^{2}}{\partial x_{i}\,\partial x_{j}}}\left[\Sigma_{ij} \,p(X ,t)\right]=0
	\] 		
	and is the unique density satisfying the condition that divergence of the flux is zero.
\end{proposition}

To ascertain how fast the transition density in the Fokker-Planck equation reaches the invariant measure, we still need to use sophisticated functional analytical methods from \cite{bolley2012convergence}.

\end{document}